\documentclass[prl,amsmath,amssymb,twocolumn,superscriptaddress]{revtex4}

\usepackage{graphicx}
\usepackage{amsmath}
\usepackage[utf8]{inputenc}
\usepackage[english]{babel}
\usepackage{amsfonts}
\usepackage{amssymb}
\usepackage{amsthm}
\usepackage{multirow}
\usepackage{mdframed}
\usepackage{color}
\usepackage{longtable}

\theoremstyle{definition}
\newtheorem*{definition*}{Definition}
\newtheorem{theorem}{Theorem}
\newtheorem*{theorem*}{Theorem}

\newcommand{\Supp}{\text{Supp}}

\begin{document}

\title{Contextuality bounds the efficiency of classical simulation of quantum processes}

\author{Angela Karanjai}
\affiliation{Centre for Engineered Quantum Systems, School of Physics, The University of Sydney, Sydney, Australia}
\author{Joel J. Wallman}
\affiliation{Institute for Quantum Computing and Department of Applied Mathematics, University of Waterloo, Waterloo, Ontario N2L 3G1, Canada}
\author{Stephen D. Bartlett}
\affiliation{Centre for Engineered Quantum Systems, School of Physics, The University of Sydney, Sydney, Australia}

\begin{abstract}
Contextuality has been conjectured to be a super-classical resource for quantum computation, analogous to the role of non-locality as a super-classical resource for communication.  We show that the presence of contextuality places a lower bound on the amount of classical memory required to simulate any quantum sub-theory, thereby establishing a quantitative connection between contextuality and classical simulability. We apply our result to the qubit stabilizer sub-theory, where the presence of state-independent contextuality has been an obstacle in establishing contextuality as a quantum computational resource.  We find that the presence of contextuality in this sub-theory demands that the minimum number of classical bits of memory required to simulate a multi-qubit system must scale quadratically in the number of qubits; notably, this is the same scaling as the Gottesman-Knill algorithm. We contrast this result with the (non-contextual) qudit case, where linear scaling is possible.
\end{abstract}

\maketitle

Quantum computers promise to outperform classical computers at several information processing tasks \cite{deutsch-jozsa,lloyd}.  However, the key properties or features of quantum theory that enable this computational advantage remain poorly understood.  Identifying such quantum resources will have immense implications for the way quantum computers are built. One strategy to isolate a quantum resource is to try to classically simulate \emph{sub-theories:} closed subsets of operations such as state preparations, unitaries and measurements of quantum systems; see~\cite{jozsa2008matchgates,bremner2010classical,aaronson2011computational,veitch,jozsa2013classical,koh2015further,pashayan2015estimating}. If the outcome statistics for a quantum sub-theory can be simulated efficiently on a classical computer, then that particular sub-theory is not resourceful for quantum computation.  Designing quantum computers that can perform super-classically will then require operations that lie outside of such sub-theories.

Several results have suggested that the ability to demonstrate contextuality might be a resource for quantum computation \cite{mbqc,aandb,magic}, analogous to how violating a CHSH inequality is a resource for non-local games.  Supporting this perspective are two quantum sub-theories---Gaussian quantum mechanics~\cite{erl}, and the odd-dimensional qudit stabilizer sub-theory~\cite{veitch}---which have been shown to be efficiently simulable classically and do not exhibit contextuality.  However, the qubit stabilizer sub-theory presents itself as potential counterexample.  This sub-theory is efficiently simulable classically and yet can exhibit contextuality such as the Peres-Mermin magic square~\cite{peres,mermin}.  The qubit stabilizer sub-theory plays a central role in the theory of fault-tolerant quantum computation, and so there is a compelling need to understand the relationship between contextuality and quantum computational resources in this sub-theory.  Several attempts to establish contextuality as a resource in this sub-theory have done so by further restricting to non-contextual subsets of operations within the qubit stabilizer sub-theory~\cite{magicqubit,rebit,juan}.  In attempting to classically simulate experiments that demonstrate contextuality, one may require additional memory to track the context, but specific results along these lines have been limited to two qubits~\cite{kleinman1, kleinman2}.

In this Letter, we demonstrate in general terms an explicit effect of any proof of Kochen-Specker contextuality on the efficiency of a classical simulation of any quantum sub-theory. We show that the presence of contextuality places a lower bound on the spatial complexity of any simulation that reproduces the quantum statistics of a sub-theory, i.e., the number of classical bits of storage required by the simulation.  We apply this result to the qubit stabilizer sub-theory and find that the presence of contextuality \emph{does} impose a cost on any classical simulation of this sub-theory: we prove that the number of classical bits required in a classical simulation of the $n$-qubit stabilizer sub-theory must grow quadratically in the number of qubits, rather than linearly as we find in simulations of stabilizer operations that do not exhibit any contextuality~\cite{veitch,juan}.

\paragraph{General framework.}
We begin by introducing the notion of a sub-theory as a restricted class of quantum experiments, and defining what is required of a classical algorithm that simulates the outcome statistics of a sub-theory.  To directly connect with the Kochen-Specker notion of contextuality, we use observables as the central objects of our framework. 

Consider a set of experiments that allow for non-destructive measurements of a set of observables $\mathbb{O}$. We require this set to be closed, meaning that if the outcomes of the measurements of two or more observables in $\mathbb{O}$ can be used to perfectly predict the outcome of a measurement of another observable, then this additional observable must also be an element of $\mathbb{O}$.  We then define a \emph{sub-theory} as the entire class of experiments with which observables in $\mathbb{O}$, and \emph{only} observables in $\mathbb{O}$, can be learnt.  Note that we do not place any restrictions on how the experimenter measures an observable or set of observables.  In general there may be multiple different experiments, that is, different sets of measurements that can be performed to learn the same information, and all of these are included in the sub-theory. A specific example of such a sub-theory is the $n$-qubit stabilizer sub-theory, consisting of all quantum circuits that allow for the measurement of observables in the $n$-qubit Pauli group.  (We note that requirement of closure excludes the type of sub-theories developed in Refs.~\cite{rebit,juan}, where certain observables cannot be directly measured but rather only inferred from other observables, introduced as a mechanism to evade contextuality.) 

We have chosen to define a sub-theory operationally in terms of the observables that one can measure.  However, it is easy to see that a quantum sub-theory is equivalently given by the set of all the pure quantum states allowed in a sub-theory defined by the full set of eigenstates of the allowed observables. For example, the $n$-qubit stabilizer sub-theory can also be defined as consisting of all the quantum circuits, where the state of the system is always described by an $n$-qubit stabilizer state. 

A sub-theory can be characterised by its statistics, that is, the  probabilities for obtaining specific outcomes when measuring an observable conditioned on previous measurement outcomes. We represent the outcomes of all previous measurements that have been performed on a system by an ordered set of observables together with their outcomes, which we call the \emph{record}.  A record is denoted by $\mathcal{R}=\{O_{i}=k_{i}\}$, where $O_{i}$ is the $i^{th}$ measured observable and $k_{i}$ its outcome. Now, let $\mathcal{M}= \{O_{j}=k_{j}\}$ be set of observable outcomes that can be jointly measured. (Operationally, joint measurability means that we have observed from experiments that the order of measurement of these observables does not matter.) The joint conditional probabilities, denoted ${\rm Pr}(\mathcal{M}|\mathcal{R})$, characterise the sub-theory.

We are interested in simulating \emph{quantum} statistics, and so we will restrict our attention to models for which the conditional probabilities ${\rm Pr}(\mathcal{M}|\mathcal{R})$ are given by the Born rule.  (Later, we can relax this restriction to include a wider class of models.)  One can imagine a wide variety of different simulation methods that reproduce these probabilities. For example, one could compute the Born rule probabilities directly using wavefunctions or density matrices.  For the qudit stabilizer sub-theory or Gaussian quantum mechanics, one could more efficiently use a generalized Gottesman-Knill representation \cite{got}, or even more efficiently, sample from their Wigner function distributions~\cite{erl, veitch}. We want to consider \emph{all} such methods, and so we now outline a general framework that can be used to analyse any classical simulation of a sub-theory, based on the ontological models framework \cite{ont}. We start by defining a (classical) state space $\Lambda$.  An element $\lambda \in \Lambda$ will be referred to as an \emph{internal} state of the model, and it encodes all the information about the system needed to determine the outcome statistics of the sub-theory. The size (cardinality) of the internal state space $\Lambda$ will determine the spatial complexity of the simulation, i.e., the number of classical bits of memory required by the simulation. This quantity is the main object of interest and will be used to evaluate the efficiency of a classical simulation. 

A simulation algorithm may be statistical in nature, and so in general we represent the system at a given time (i.e., conditioned on a given record) using a distribution over $\Lambda$. Similarly, $\lambda$ might only probabilistically determine the outcomes of measurements of observables.  We introduce two probability density functions defined on the internal state space of the model. The first, $\mu_{\mathcal{R}}(\lambda)$, is the probability of the internal state being $\lambda\in\Lambda$, given the record $\mathcal{R}$. The second, $\xi_{\mathcal{M}}(\lambda)$, is the probability of the outcomes of the observables $\mathcal{M}$, given that the system is in the internal state $\lambda$.  The classical simulation will then reproduce statistics according to
${\rm Pr}(\mathcal{M}|\mathcal{R})= \sum_{\lambda \in \Lambda}\mu_{\mathcal{R}}(\lambda)\xi_{\mathcal{M}}(\lambda)$.  

We need one additional ingredient to our general framework:  an update map, used to describe how the internal state (and corresponding probability distribution) changes as the result of a measurement.  (Our framework allows for transformations other than measurements, but we restrict our attention to measurement update for clarity.)  Let $\Gamma_{\mathcal{M}}(\lambda'|\lambda)$ be the conditional probability of the state of the system being described by the internal state $\lambda'$ given the measurement $\mathcal{M}$ is performed on a system originally in internal state $\lambda$. This is henceforth known as the \emph{update map} and if the operation $\mathcal{M}$ updates record $\mathcal{R}$ to $\mathcal{R'}$ (by appending $\mathcal{M}$ to the record), then we have 
$\mu_{\mathcal{R'}}(\lambda')=\sum_{\lambda\in\Lambda}\Gamma_{\mathcal{M}}(\lambda'|\lambda)\mu_{\mathcal{R}}(\lambda)$.  In summary, any classical simulation can be defined by four objects, $(\Lambda, \{ \mu_{\mathcal{R}}\},\{\xi_{\mathcal{M}}\},\{\Gamma_{\mathcal{M}}\})$~\footnote{ Our results do not apply to simulations where the state of the system does not have a representation independent of the measurements in the circuit, such as algorithms which estimate circuit probabilities merely by computing inner products.}.

\paragraph{Minimal requirements for simulating a quantum sub-theory.}
While we are considering classical simulations that reproduce the Born rule probabilities of a quantum sub-theory, to prove our lower bounds on efficiency we will only make use of two conditions, which are in general weaker.

Consider an experiment where the record $\mathcal{R}$ describes the system in the simulation, and $\rho$ is the corresponding quantum description. After some measurement or a set of measurements, the record describing the state is now $\mathcal{R'}$, which is quantum mechanically represented as $\rho'$.  In the quantum description of the experiment, the measurement will correspond to a map that updates the state of the system from $\rho$ to $\rho'$. In a simulation, the measurement will correspond to an update map that takes any internal state in the support of $\mu_{\mathcal{R}}$ to one in the support of $\mu_{\mathcal{R'}}$. There may in general be multiple records that have the same quantum description \cite{spekkens}, that is, the records in the set $\{\mathcal{R}_{i}\}$ can all be represented by the same quantum state $\rho$. From here on, we shall use a compact notation where  $\Supp(\mu_{\rho})\equiv \cup_{i}\Supp(\mu_{\mathcal{R}_i})$ to compare the simulation with the quantum representation. We can now use this notation to give our condition of state update as 
\begin{equation}
\tag{State Update}
\label{update}
\begin{split}
\mathcal{M}:\rho \to & \rho'\\
\implies \Gamma_{\mathcal{M}}: \Supp(\mu_{\rho})\to & \Supp(\mu_{\rho'})\\
\end{split}
\end{equation}

Our second condition is a requirement on distinguishability.  If it is possible within the sub-theory to perform a measurement that perfectly distinguishes between two different records, $\mathcal{R}$ and $\mathcal{R'}$, then we say $\mathcal{R}$ and $\mathcal{R'}$ are \emph{single-shot distinguishable} (SSD).  Clearly, two records that are SSD must be represented disjointly in any simulation. Within a quantum description, any two records represented by projectors onto orthogonal subspaces are SSD. Thus, we have our second condition:
\begin{equation}
\tag{SSD}
\label{SSD}
\Supp(\mu_{\rho})\cap \Supp (\mu_{\sigma}) = \emptyset \ \ \text{if} \ {\rm Tr}(\rho\sigma)=0
\end{equation}
In other words, no internal state can be in the support of two records that are SSD. We shall now use \emph{only} \eqref{update} and \eqref{SSD} to prove restrictions on simulations of quantum sub-theories.

\paragraph{Partitioning measurements.}
We now turn to our central question: What is the minimum number of classical bits of memory required to simulate a quantum sub-theory? We shall answer this question by bounding the minimum size of the classical state space required for any simulation of the sub-theory satisfying \eqref{update} and \eqref{SSD}. Having a unique $\lambda \in \Lambda$ for every quantum state in the sub-theory would in general be highly inefficient, and most simulations make use of internal state spaces where a given $\lambda$ can lie in the support of multiple distributions describing distinct records.The SSD condition ensures non-intersecting support of distributions for any pair of records associated with orthogonal states, giving a lower bound on the cardinality of $\Lambda$ equal to the number of orthogonal states in a sub-theory.  This lower bound is generic for quantum sub-theories that do not exhibit contextuality.  (As an example, for the qudit stabilizer sub-theory, this bound requires a number of classical dits that scales linearly in the number of qudits, and this bound is saturated by the discrete Wigner function representation~\cite{veitch}.)  However, as we now show, in certain cases non-orthogonal quantum states must have  non-intersecting support of distributions, giving  stronger bounds for sub-theories that exhibit contextuality.   We introduce the concept of a \emph{partitioning measurement} for a set of quantum states. This will be our central tool for obtaining these stronger bounds. 

Let $s=\{\rho_{i}\}$ be a set of quantum states, and $s_{O,k}=\{\sigma_{k,i}\}$ be the set of post-measurement quantum states after a projective measurement of observable $O$ with outcome $k$.  We say that $O$ is a \emph{partitioning measurement} for $s$ if the set of post measurement states $s_{O,k}$ contains at least one pair of orthogonal quantum states for every outcome $k$.

\begin{theorem}[Partitioning]
\label{first}
If there exists a partitioning measurement for a set of quantum states $s=\{\rho_{i}\}$ within the sub-theory, then
$\cap_{i}\Supp(\mu_{\rho_{i}}) =\emptyset$ for any simulation of this sub-theory.
\end{theorem}

\begin{proof}
The proof is by contradiction.  Assume that there exists an internal state $\lambda \in \cap _i \Supp(\mu_{\rho_{i}})$. According to \eqref{update}, as a result of the measurement of $O$ with outcome $k$, $\lambda$ must be mapped to $\lambda' \in \cap _i \Supp(\mu_{\sigma_{k,i}})$. However, $s_{O,k}$ contains a pair of orthogonal states, and so $\cap _i \Supp(\mu_{\sigma_{k,i}})=\emptyset$ according to \eqref{SSD}. Noting that for the relevant update map $\Gamma_{O,k}$, $\lambda$ must lie in the pre-image, that is, $\lambda\in \Gamma_{O,k}^{-1}[\cap _i \Supp(\mu_{\sigma_{k,i}})]$
 gives us the required contradiction. 
\end{proof}

As an example of a partitioning measurement for a set of non-orthogonal quantum states, consider the set of two-qubit states used in the Pusey-Barrett-Rudolph (PBR) theorem~\cite{PBR}, $s=\{|00\rangle, |0{+}\rangle,|{+}0\rangle,|{+}{+}\rangle\}$, where $|0\rangle$ and $|{+}\rangle$ are the $+1$ eigenstates of the Pauli $Z$ and $X$ operators respectively. A partitioning measurement for this set is  measurement of the observable $YY$ ($=Y\otimes Y$), which has outcomes $YY=\pm1$. For $YY=+1$, the states $|00\rangle$ and $|{+}{+}\rangle$ in $s$ are mapped to two orthogonal quantum states; for $YY=-1$, the states $|0{+}\rangle$ and $|{+}{0}\rangle$ in $s$ are mapped to two orthogonal quantum states. That is, the set of post-measurement quantum states includes an orthogonal pair of states for either measurement outcome. Thus, according to Theorem \ref{first}, $\Supp(\mu_{\scriptscriptstyle{|00\rangle}})\cap \Supp(\mu_{\scriptscriptstyle{|0{+}\rangle}})\cap\Supp(\mu_{\scriptscriptstyle{|{+}0\rangle}})\cap\Supp(\mu_{\scriptscriptstyle{|{+}{+}\rangle}})=\emptyset$. Thus Theorem \ref{first} provides a proof of the PBR theorem without the assumption of preparation independence.

\paragraph{Contextuality in a sub-theory.}
When a partitioning measurement exists for a set of quantum states, Theorem \ref{first} tells us that no single internal state can be in the common support of all the distributions corresponding to the quantum states in the set.  As such, the existence of a partitioning measurement provides a lower bound on the number of internal states needed in a simulation.  What is needed, then, is a method of finding partitioning measurements in general sub-theories.  As we now show,  a partitioning measurement exists whenever one is able to construct a proof of contextuality. With this, one can use existing proofs of contextuality (e.g.,~\cite{csw, belen, sheaf, sam,kai}) to identify partitioning measurements for sub-theories of interest. 

A \emph{non-contextual value assignment} (NCVA) of a set of  operators $\mathbb{\hat{O}}$ is a function $\nu:\mathbb{\hat{O}}\to\mathbb{C}$ such that $\nu(\hat{O}_{\alpha})$ is an eigenvalue of $\hat{O}_{\alpha}$ and $\nu(\hat{O}_{\alpha})\cdot\nu(\hat{O}_{\beta})=\nu(\hat{O}_{\alpha}\cdot\hat{O}_{\beta})$ if $\hat{O}_{\alpha}$ and $\hat{O}_{\beta}$ commute.  For certain sets of operators that represent observables in quantum mechanics, it can be shown that a NCVA cannot exist~\cite{KS}.  This is known as a proof of \emph{contextuality}.

Consider a set of quantum states $s=\{\rho_{i}\}$.  Let  $\mathbb{O}_{s}$ be the set of all observable operators that have at least one eigenstate in $s$. Now consider the set of the products of the commuting elements in $\mathbb{O}_{s}$, that is,  $\mathbb{\tilde{O}}_{s}=\{O_{\alpha}O_{\beta}|O_{\alpha},O_{\beta}\in \mathbb{O}_{s}\ \text{and} \ [O_{\alpha},O_{\beta}]=0\}$.

 \begin{theorem}
  \label{main1}
If $O_{P} \in \mathbb{\tilde{O}}_{s}$ and $\mathbb{O}_{s}\cup O_{P}$ allows a proof of contextuality then $O_{P}$ is a partitioning measurement for the set $s$.
\end{theorem}
 
\begin{proof}
If $s$ contains an orthogonal pair of states, then it follows from the definition of $\mathbb{O}_{s}$, that $\exists \ O_{P}\in \mathbb{O}_{s}$, which is a partitioning measurement for $s$.  Let $\nu:\mathbb{{O}}_{s}\to\mathbb{C}$ be the function that assigns eigenvalues to the operators in $\mathbb{O}_{s}$. If $s$ contains no orthogonal pairs of states then all the eigenstates of an observable ${O}_{\alpha}\in \mathbb{O}_{s}$ in $s$ must have the same eigenvalue $\nu({O}_{\alpha})$. If the observable corresponding to the operator ${O}_{\alpha}{O}_{\beta}$ is measured, where ${O}_{\alpha}$ and ${O}_{\beta}$ commute, and the outcome of the measurement is $k(O_{\alpha}O_{\beta})$, then all eigenstates of ${O}_{\alpha}$ will be mapped to a state which lies in the intersection of the eigenspaces of ${O}_{\alpha}$ and ${O}_{\beta}$ with eigenvalues $\nu({O}_{\alpha})$ and $\frac{k(O_{\alpha}O_{\beta})}{\nu({O}_{\alpha})}$ respectively. Similarly, all eigenstates of ${O}_{\beta}$ will be mapped to a state which is an eigenstate of ${O}_{\beta}$ and ${O}_{\alpha}$ with eigenvalues $\nu({O}_{\beta})$ and $\frac{k(O_{\alpha}O_{\beta})}{\nu({O}_{\beta})}$ respectively. In order for the post measurement states to \emph{all} be mutually non-orthogonal, we require that all the eigenstates of an operator have the same eigenvalue, and therefore $\nu({O}_{\alpha})\nu({O}_{\beta})=k(O_{\alpha}O_{\beta}) \forall O_{\alpha},O_{\beta}\in \mathbb{O}_{s}$ thus proving the result.
\end{proof}

As an example, note that if $s$ is the set of PBR states, then $YY\in \tilde{\mathbb{O}}_{s}$ and $\mathbb{O}_{s}\cup YY$ are all the observables that make up the Peres-Mermin magic square.

It is easy to see that if a set of states $s$ is allowed in a sub-theory, then measurement of all the observables in $\mathbb{\tilde{O}}_{s}$ are allowed within the sub-theory. Thus theorem \ref{main1} can be used to find partitioning measurements for sets of states within the sub-theory. We can now describe our method for lower bounding the size of the internal state space of any classical simulation of a quantum sub-theory. Consider a quantum sub-theory defined by $S$, the set of all allowed pure quantum states. Using theorem \ref{main1}, one can find the cardinality of the largest set of pure quantum states that do not have a partitioning measurement in the sub-theory; denote this number by $m$. Theorem \ref{first} implies that a single internal state can be in the support of at most $m$ pure quantum states. Thus a lower bound on the number of internal states required by a classical simulation of the sub-theory is $|S|/m$.

\paragraph{Qubit stabilizer sub-theory.}
As an example of this approach, we now bound the minimum number of internal states required in a simulation of the qubit stabilizer sub-theory.
The $n$-qubit Pauli group $\mathcal{P}_n$ is the group made up of all tensor products of $n$ Pauli matrices, together with a multiplicative factor of $\pm i$ and $\pm 1$. The Hermitian elements of this group represent observables; this set is closed and thus defines a sub-theory. For two qubits, we find that the cardinality of the largest set of pure stabilizer states that satisfy the condition necessary for having a non-empty overlap, is 5. An example of such a set is 
\{$[XI,IX,XX]$, $[ZI,IZ,ZZ]$, $[XI,IZ,XZ]$, $[YI,IX,YX]$, $[ZZ,YX,XY]$\}, where the quantum states are defined in terms of their stabilisers. So, for any set of 6 or more two-qubit stabilizer states, there exists a partitioning measurement. The following theorem extends this bound to an $n$-qubit system:

\begin{theorem}
\label{maintheorem}
Let $s_{n}=\{|\psi_i\rangle\}$ be a set of pure $n$-qubit stabilizer states, where $n\geq2$. Then $\cap_{i}\Supp(\mu_{|\psi_i\rangle})= \emptyset$ if $|s_{n}|>3^{n-2}5$.
\end{theorem}

\begin{proof}
We prove this theorem iteratively.  
Consider a set of pure $k+1$ qubit stabilizer states, where the $k+1^\text{th}$ qubit is unentangled with the rest of the system. Let $\tilde{s}_{k+1}=\{|\phi_i\rangle|\nu\rangle\}$ be such a set, where $|\nu\rangle$ is a single qubit stabilizer state. One can see that if there exists a partitioning measurement for the set $s_{k}=\{|\phi_i\rangle\}$, then the same measurement will also be a partitioning measurement for $\tilde{s}_{k+1}$. Thus the cardinality of a set of the form $\tilde{s}_{k+1}$ that does not allow a partitioning measurement is upper bounded by the cardinality of the largest set of $k$-qubit stabilizer states that does not allow a partitioning measurement.
Now consider the set $s_{k+1}=\{|\psi_i\rangle\}$, a set of pure $k+1$ qubit stabilizer states. A measurement of a Pauli $P$ only on the last qubit with outcome $P=\nu$ would map the states in $s_{k+1}$ to a set of post measurement states of the form $\tilde{s}_{k+1}$ given above. Note $|\tilde{s}_{k+1}|$ may be smaller than $|s_{k+1}|$, as some states in $s_{k+1}$ may map to the same post-measurement state.  However, one can see that for any set of pure non-orthogonal states larger than 3, there exists a single qubit Pauli measurement that will map the states to a set of post measurement states with more than one state. Thus we have $3|\tilde{s}_{k+1}|\geq|s_{k+1}|$, which in turn implies $3|s_{k}|\geq|s_{k+1}|$.
As we have (by brute force search) that for any $|s_{2}| >5$, there exists a partitioning measurement, we have proved the theorem.
\end{proof}

For the $n$-qubit stabilizer sub-theory, the total number of pure states is $|S|=2^n\prod^{n}_{j}(2^j+1)$ and theorem \ref{maintheorem} shows that the largest number of states that can be represented by the same internal state is $m=5\times3^{(n-2)}$. Thus the minimum number of classical bits required to simulate the sub-theory is $\log_{2}(|S|/m) \approx \frac{1}{2}n(n-1)$. We compare this result with that of the Gottesman-Knill algorithm~\cite{got} (a specific classical simulation), which uses $n(2n+1)$ classical bits. Our result shows that no classical algorithm can scale better than the Gottesman-Knill algorithm, i.e., the number of classical bits must scale quadratically with the number of qubits.

\paragraph{Discussion.}
Our result provides insight into the differences in classical simulation of the qubit and odd-dimensional qudit stabilizer sub-theories.  The internal states in the Gottesman-Knill algorithm have a one-to-one correspondence with each quantum stabilizer state; as a result, this algorithm is efficient but the number of classical bits of memory required scales quadratically (rather that linearly) in the number of qubits.  From Theorem~\ref{main1}, we can see that it is the presence of contextuality that lower bounds the number of quantum states that can be represented by a single internal state, thereby ruling out linear-scaling algorithms such as algorithms that sample from a discrete Wigner function (as is possible with the non-contextual odd-dimensional qudit stabilizer sub-theory~\cite{veitch}) proving the Gottesman-Knill algorithm to be asymptotically optimal for the qubit stabilizer sub-theory. 

We have illustrated our result using the qubit stabilizer sub-theory, in part because of its key role in the theory of quantum computing, but also because the role of contextuality in this sub-theory is well known. Our result could be used to find lower bounds on the spatial complexity of simulations of other sub-theories that are relevant to quantum computing include matchgates~\cite{jozsa2008matchgates}, IQP circuits~\cite{bremner2010classical}, boson sampling circuits~\cite{aaronson2011computational} and thus provide a motivation to investigate the role of contextuality in these sub-theories. These lower bounds can then be compared to the best known classical algorithm at the time. Our approach may also be generalizable to classes of circuits that are not closed, such those of Ref.~\cite{bravyi,jozsa2013classical,koh2015further}.

In proving our results, we note that classical simulations of quantum sub-theories were only required to satisfy \eqref{update} and \eqref{SSD}, both of which are possibilistic requirements on the simulation (rather than necessarily reproducing the full probabilities of the sub-theory).  That is, we only require that the simulation does not predict any event that is impossible according to quantum theory.  Our results do assume that these (im)possibilities are reproduced exactly in the simulation, although we note that the methods of Ref.~\cite{pashayan2017estimation} may be used to extend our results to classical simulation up to small error. Additionally, because the space complexity of a simulation provides a lower bound for the time complexity of the simulation, our result also provides a lower bound for the run time of a classical simulation. 

Finally, we note that research in quantum foundations has sought bounds on the overlap of representations of quantum states on any classical state space in terms of their quantum overlap \cite{leifer2013maximally,epistemic1, epistemic2}.  These bounds typically require consideration of the full set of states and operations for a Hilbert space of a fixed dimension, and focus on pairwise and tri-partite overlap. Our results offers a new approach to these questions for any operational sub-theory of quantum mechanics.

\paragraph{Acknowledgments}
We thank Christopher Chubb and Mordecai Waegell for discussions.  This work is supported by the Australian Research Council (ARC) via the Centre of Excellence in Engineered Quantum Systems (EQuS) project number CE170100009.

\end{document}